\documentclass{article} % For LaTeX2e
\usepackage{fullpage}
\usepackage{url}

\title{The Pseudo-Dimension of Near-Optimal Auctions}

\author{ JAMIE MORGENSTERN\thanks{Carnegie Mellon University, {\tt
      jamiemmt@cs.cmu.edu}, supported in part by Simons Graduate
    Fellowship in Theoretical Computer Science and NSF grants
    CCF-1415460. Part of work done while a visiting student at
    Stanford University.} \and TIM ROUGHGARDEN \thanks{Stanford
    University{\tt tim@cs.stanford.edu}, Roughgarden was supported in
    part by NSF Award CCF-1215965 and an ONR PECASE Award.}}

\usepackage{enumitem}

\usepackage{amsmath}
\usepackage{amssymb,amsfonts}
\usepackage{theorem}
\usepackage{xspace}
\usepackage[numbers]{natbib}
\usepackage{graphicx}
\usepackage{url}
\newcommand{\ind}{\tau}

\newcommand{\poly}{\textrm{poly}}

\renewcommand{\t}{t}
\newcommand{\lev}[2]{\ell_{#1, #2}}
\renewcommand{\v}{{\mathbf {v}}}
\newcommand{\vals}{{\mathbf {v}}}
\newcommand{\val}{{v}}
\newcommand{\winner}{{i^{*}}}
\newcommand{\mwinner}{{i^{'}}}
\newcommand{\winners}{{\mathcal{I}^{*}}}
\newcommand{\mwinners}{{\mathcal{I}^{'}}}
\newcommand{\sse}{{\subseteq}}

\newtheorem{theorem}                            {Theorem}[section]
\newtheorem{lemma}              [theorem]       {Lemma}

\newtheorem{corollary}          [theorem]       {Corollary}
\newtheorem{prop}               [theorem]       {Proposition}

{\theorembodyfont{\rmfamily} }
{\theorembodyfont{\rmfamily} }
{\theorembodyfont{\rmfamily} \newtheorem{remark}                [theorem]
{Remark}}
{\theorembodyfont{\rmfamily} \newtheorem{example}               [theorem]
{Example}}
{\theorembodyfont{\rmfamily} }
\theoremstyle{break}
{\theorembodyfont{\rmfamily} }

\usepackage{subfig}
\newcommand{\D}{\mathcal{D}}
\newcommand{\R}{\mathbb{R}}
\newcommand{\eps}{\epsilon}
\newcommand{\F}{F}
\newcommand{\Ct}{\mathcal{C}_t}
\newcommand{\C}{\mathcal{C}}
\newcommand{\G}{\mathcal{G}}
\newcommand{\dom}{\mathcal{Q}}

\newcommand{\bounded}{\frac{1}{\epsilon} + \log_{1 + \epsilon}H}
\newcommand{\VC}{\mathcal{VC}}
\newcommand{\err}{\textrm{err}}

\newcommand{\pd}{\textrm{d}}
\newcommand{\X}{\mathcal{X}}
\newcommand{\A}{\mathcal{A}}
\newcommand{\M}{\mathcal{M}}
\newcommand{\E}{\mathbb{E}}
\newcommand{\pr}{\mathbb{P}}

\newcommand{\vv}[2]{\phi_{#1}(#2)}
\newcommand{\vvi}[2]{\phi^{-1}_{#1}(#2)}

\newenvironment{proof}{\noindent {\em {Proof:}}}{$\blacksquare$\vskip \belowdisplayskip}

\newenvironment{prevproof}[2]{\noindent {\em {Proof of                          
{#1}~\ref{#2}:}}}{$\blacksquare$\vskip \belowdisplayskip}

\newcommand{\rev}{\textrm{Rev}}

\begin{document}

\maketitle

\begin{abstract}
  This paper develops a general approach, rooted in statistical
  learning theory, to learning an approximately revenue-maximizing
  auction from data.  We introduce {\em $t$-level auctions} to
  interpolate between simple auctions, such as welfare maximization
  with reserve prices, and optimal auctions, thereby balancing the
  competing demands of expressivity and simplicity.  We prove that
  such auctions have small representation error, in the sense that for
  every product distribution $F$ over bidders' valuations, there
  exists a $t$-level auction with small~$t$ and expected
  revenue close to optimal.  We show that the set of $t$-level
  auctions has modest pseudo-dimension (for polynomial~$t$) and
  therefore leads to small learning error.  One consequence of our
  results is that, in arbitrary single-parameter settings, one can
  learn a mechanism with expected revenue arbitrarily close to optimal
  from a polynomial number of samples.
\end{abstract}

\section{Introduction}\label{sec:intro}

In the traditional economic approach to identifying a
revenue-maximizing auction, one first posits a prior distribution over
all unknown information, and then solves for the auction that
maximizes expected revenue with respect to this distribution.  The
first obstacle to making this approach operational is the difficulty
of formulating an appropriate prior.  The second obstacle is that,
even if an appropriate prior distribution is available, the
corresponding optimal auction can be far too complex and unintuitive
for practical use.  This motivates the goal of identifying auctions
that are ``simple'' and yet nearly-optimal in terms of expected
revenue.

In this paper, we apply tools from learning theory to address both of
these challenges.  In our model, we assume that bidders' valuations
(i.e., ``willingness to pay'') are drawn from an unknown distribution
$F$.  A learning algorithm is given i.i.d.\ samples from $F$.  For
example, these could represent the outcomes of comparable transactions
that were observed in the past.  The learning algorithm suggests an
auction to use for future bidders, and its performance is measured by
comparing the expected revenue of its output auction to that earned by
the optimal auction for the distribution~$F$.

The possible outputs of the learning algorithm correspond to some set
$\C$ of auctions.  We view $\C$ as a design parameter that can be
selected by a seller, along with the learning algorithm.  A central
goal of this work is to identify classes $\C$ that balance
representation error (the amount of revenue sacrificed by restricting
to auctions in $\C$) with learning error (the generalization error
incurred by learning over $\C$ from samples).  That is, we seek a set
$\C$ that is rich enough to contain an auction that closely
approximates an optimal auction (whatever $F$ might be), yet simple
enough that the best auction in $\C$ can be learned from a small
amount of data.  Learning theory offers tools both for rigorously
defining the ``simplicity'' of a set $\C$ of auctions, through
well-known complexity measures such as the pseudo-dimension, and for
quantifying the amount of data necessary to identify the approximately
best auction from $\C$.  Our goal of learning a near-optimal auction
also requires understanding the representation error of different
classes $\C$; this task is problem-specific, and we develop the
necessary arguments in this paper.

\subsection{Our Contributions}

The primary contributions of this paper are the following.  First, we
show that well-known concepts from statistical learning theory can be
directly applied to reason about learning from data an approximately
revenue-maximizing auction.  Precisely, for a set $\C$ of auctions and
an arbitrary unknown distribution $F$ over valuations in $[1,H]$,
$O(\tfrac{H^2}{\eps^2}\pd_{\C} \log \tfrac{H}{\eps})$ samples from $F$
are enough to learn (up to a $1-\eps$ factor) the best auction in
$\C$, where $\pd_{\C}$ denotes the {\em pseudo-dimension} of the set
$\C$ (defined in Section~\ref{sec:learning}).
%This guarantee holds even if bidders'
%  valuations are correlated.\footnote{If bidders' valuations are
%    independent and satisfy the monotone hazard rate (MHR) assumption
%    (but are not necessarily identical, or bounded), we can use results of
%    \cite{cai2011extreme} to remove the dependence on~$H$ (see
%    Section~\ref{sec:mhr}).}  As one would expect, the learning
 % algorithm simply chooses the auction of $\C$ with the highest
 % average revenue on the samples.
Second,
we introduce the class of {\em $t$-level auctions}, to
  interpolate smoothly between simple auctions, such as welfare
  maximization subject to individualized reserve prices (when $t=1$),
  and the complex auctions that can arise as optimal auctions
  (as $t \rightarrow \infty$).
Third, we prove that
%\item We prove that, for an arbitrary single-parameter 
in quite general auction settings with
  $n$ bidders, the pseudo-dimension of the set of $t$-level auctions is
  $O(nt \log nt)$.
%, despite there being an uncountable number of such
%  auctions.
Fourth,
we quantify the number~$t$ of levels required for the set of
  $t$-level auctions to have low representation error, with respect to
  the optimal auctions that arise from arbitrary product distributions
  $F$.  For example, for single-item auctions and several
  generalizations thereof, if 
  $t=\Omega(\tfrac{H}{\eps})$, then for every product distribution $F$
  there exists a $t$-level auction with expected revenue at least
  $1-\eps$ times that of the optimal auction for $F$.  
%This same bound
%  holds more generally in matroid environments.  (And with MHR
%  distributions, $t$ depends only on $\eps$.)  For general
%  single-parameter settings (not necessarily downward-closed) with $n$
%  bidders, we prove that if $t = \Omega(\tfrac{Hn^2}{\eps})$, then the
%  class of $t$-level auctions contains an auction with revenue within
%  an additive $\epsilon$ of the optimal auction for $\F$.
%\end{enumerate}

  In the above sense, the ``$t$'' in $t$-level auctions is a tunable
  ``sweet spot'', allowing a designer to balance the competing demands
  of expressivity (to achieve near-optimality) and simplicity (to
  achieve learnability).  For example, given a fixed amount of past
  data, our results indicate how much auction complexity (in the form
  of the number of levels $t$) one can employ without risking
  overfitting the auction to the data.

Alternatively, given a target approximation factor $1-\eps$, our results
give sufficient conditions on $t$ and consequently on the number of
samples needed to achieve this approximation factor.  The resulting
sample complexity upper bound has polynomial dependence on $H$,
$\eps^{-1}$, and the number $n$ of bidders.
Known results~\cite{AB,CR14} imply that any method of learning a
$(1-\eps)$-approximate auction from samples must have sample complexity with
polynomial dependence on all three of these parameters, even for
single-item auctions.
% with independent bidder valuations.

\subsection{Related Work}\label{sec:related}

The present work shares much of its spirit and
high-level goals with~\citet{balcan2008reducing}, who
%Ten years ago, the prescient paper of~\citet{balcan2008reducing}
%(following~\citet{B+03} and~\citet{blum2005near}) showed that
%statistical learning theory provides a framework for answering all of
%these questions, and 
proposed applying 
%the framework 
statistical learning theory to the design of
near-optimal auctions.  
%The present work revisits this idea under a
%stronger assumption --- replacing the prior-free assumption used
%in~\citet{balcan2008reducing} with a distributional model --- which
%gives results in much more general auction environments and tighter
%connections to both learning theory and Bayesian-optimal auction
%theory.
The first-order
difference between the two works is that our work assumes bidders'
valuations are drawn from an unknown distribution, while
~\citet{balcan2008reducing} study the more demanding ``prior-free''
setting. Since no auction can achieve near-optimal revenue
ex-post,
%~\cite{G+06},~
\citet{balcan2008reducing} define their revenue
benchmark with respect to a set $\G$ of auctions on each input $\vals$
as the maximum revenue obtained by any auction of $\G$ on $\vals$.
The idea of learning from samples enters the work
of~\citet{balcan2008reducing} through the internal randomness of their
partitioning of bidders, rather than through an exogenous distribution
over inputs (as in this work). Both our work and theirs requires
polynomial dependence on $H, \tfrac{1}{\eps}$: ours in terms of a
necessary number of samples, and theirs in terms of a necessary number
of bidders; as well as a measure of the complexity of the class $\G$
(in our case, the pseudo-dimension, and in theirs, an analagous
measure).  The primary improvement of our work over of the results
in~\citet{balcan2008reducing} is that our results apply for single
item-auctions, matroid feasibility, and arbitrary single-parameter
settings (see Section~\ref{sec:prelim} for definitions); while their
results apply only to single-parameter settings 
of unlimited 
supply.\footnote{See~\citet{balcan2007CMUtechreport} for an extension
to the case of a large finite supply.} We also view as a feature the
fact that our sample complexity upper bounds can be deduced directly
from well-known results in learning theory --- we can focus instead on
the non-trivial and problem-specific work of bounding the
pseudo-dimension and representation error of well-chosen auction
classes.

\citet{elkind2007} also considers a similar model to ours, but only
for the special case of single-item auctions.  While her proposed
auction format is similar to ours, our results cover the far more
general case of arbitrary single-parameter settings and and non-finite
support distributions; our sample complexity bounds are also better
even in the case of a single-item auction (linear rather than
quadratic dependence on the number of bidders).  On the other hand,
the learning algorithm in \cite{elkind2007} (for single-item auctions)
is computationally efficient, while ours is not.

\citet{CR14} study single-item auctions with $n$ bidders with
valuations drawn from independent (not necessarily identical)
``regular'' distributions (see Section~\ref{sec:prelim}), and prove
upper and lower bounds (polynomial in $n$ and $\eps^{-1}$) on the
sample complexity of learning a $(1-\eps)$-approximate auction.  While
the formalism in their work is inspired by learning theory, no formal
connections are offered; in particular, both their upper and lower
bounds were proved from scratch. Our positive results include
single-item auctions as a very special case and, for bounded or MHR
valuations, our sample complexity upper bounds are much better than
those in~\citet{CR14}.

\citet{huang2014making} consider learning the optimal price from
samples when there is a single buyer and a single seller; this problem
was also studied implicitly in~\cite{dhangwatnotai2010revenue}.  Our
general positive results obviously cover the bounded-valuation and MHR
settings in \cite{huang2014making}, though the specialized analysis in
\cite{huang2014making} yields better (indeed, almost optimal) sample
complexity bounds, as a function of $\eps^{-1}$ and/or $H$.

\citet{medina2014learning} show how to use a combination of the
pseudo-dimension and Rademacher complexity to measure the sample
complexity of selecting a single reserve price for the  VCG mechanism
to optimize revenue.  In our notation, this corresponds to analyzing a
single set~$\C$ of auctions (VCG with a reserve).
\citet{medina2014learning} do not address the expressivity
vs.\ simplicity trade-off that is central to this paper.

\citet{dughmi2014sampling} also study the sample complexity of learning good
auctions, but their main results are negative (exponential sample
complexity), for the difficult scenario of multi-parameter settings.
(All settings in this paper are single-parameter.)

Our work on $t$-level auctions also contributes to the literature on
simple approximately revenue-maximizing auctions (e.g.,
\cite{chawla2007algorithmic,hartline2009simple,
  chawla2010multi,devanur2011prior,
  RTY12,yao2015soda,moshematt2014focs}).  Here, one takes the
perspective of a seller who knows the valuation distribution $\F$ but
is bound by a ``simplicity constraint'' on the auction deployed,
thereby ruling out the optimal auction.  Our results that bound the
representation error of $t$-level auctions
(Theorems~\ref{thm:rep-bounded},~\ref{thm:rep-mhr},~\ref{thm:rep-matroid},
and~\ref{thm:rep-param}) can be interpreted as a principled way to
trade off the simplicity of an auction with its approximation
guarantee.  While previous work in this literature generally left the
term ``simple'' safely undefined, this paper effectively proposes the
pseudo-dimension of an auction class as a rigorous and quantifiable
simplicity measure.

\section{Preliminaries}\label{sec:prelim}

This section reviews useful terminology and notation standard in
Bayesian auction design and learning theory.

\paragraph{Bayesian Auction Design}\label{sec:bayesian}

We consider \emph{single-parameter settings} with $n$ bidders.
This means that each bidder has a single unknown parameter, its {\em
  valuation} or willingness to pay for ``winning.''
(Every bidder has value~0 for losing.)
A setting is specified by a
collection $\mathcal{X}$ of subsets of $\{1,2,\ldots,n\}$; each such
subset represent a collection of bidders that can simultaneously
``win.''  For example, in a setting with $k$ copies of an item, where
no bidder wants more than one copy, $\mathcal{X}$ would be all subsets
of $\{1,2,\ldots,n\}$ of cardinality at most~$k$.

A generalization of this case, studied in Section~\ref{sec:matroid},
is {\em matroid} settings.  These satisfy: (i) whenever $X \in \X$ and
$Y \subseteq X$, $Y \in \X$; and (ii) for two sets $|I_1| < |I_2|$,
$I_1, I_2 \in \X$, there is always an augmenting element
$i_2 \in I_2 \setminus I_1$ such that $I_1 \cup \{i_2\}\in \X$, $\X$.
Section~\ref{sec:single-param} also considers arbitrary
single-parameter settings, where the only assumption is that
$\emptyset \in \X$.  To ease comprehension, we often illustrate our
main ideas using single-item auctions (where $\X$ is the singletons
and the empty set).

We assume bidders' valuations are drawn from the continuous joint
cumulative distribution $F$.  Except in the extension in
Section~\ref{sec:mhr}, we assume that the support of $F$ is limited to
$[1,H]^n$.  As in most of optimal auction
theory~\cite{myerson1981optimal}, we usually assume that $F$ is a
product distribution, with $F= F_1\times F_2 \times \ldots \times F_n$
and each $v_i \sim F_i$ drawn independently but not identically.
%; we will
%mention explicitly where independence is needed for some results. 
The \emph{virtual value} of bidder $i$ is denoted by
$\phi_i(v_i) = v_i - \frac{1-F_i(v_i)}{f_i(v_i)}$.  A distribution
satisfies the \emph{monotone-hazard rate (MHR)} condition if
$f_i(v_i)/(1-F_i(v_i))$ is nondecreasing; intuitively, if its tails
are no heavier than those of an exponential distribution.  In a
fundamental paper, \cite{myerson1981optimal} proved that when every
virtual valuation function is nondecreasing (the ``regular'' case),
the auction that maximizes expected revenue for $n$ Bayesian bidders
chooses winners in a way which maximizes the sum of the virtual values
of the winners. This auction is known as Myerson's auction, which we
refer to as $\M$. The result can be extended to the general,
``non-regular'' case by replacing the virtual valuation functions by
``ironed virtual valuation functions.''  The details are
well-understood but technical; see~\citet{myerson1981optimal} and
%When $ \phi_i$ is not weakly increasing (i.e., $F_i$ is
%not {\em regular}), the optimal auction \emph{ironed} virtual value $\bar\phi_i(v_i)$,
%which is monotone, can be used instead (for details,
\citet{hartline2015} for details.
%). For settings where agents are not assumed
%to be regular, running $\M$ with respect to the ironed virtual
%valuation functions maximizes the expected revenue over all
%incentive-compatible mechanisms.

\paragraph{ Sample Complexity, VC Dimension, and the
  Pseudo-Dimension}\label{sec:learning}

This section reviews several well-known definitions from learning
theory. Suppose there is some domain $\dom$, and let $c$ be some
unknown target function $c: \dom \to \{0,1\}$.
Let $\D$ be an unknown distribution over $\dom$. We wish to understand
how many labeled 
samples $(x, c(x))$, $x\sim \D$, are necessary and sufficient to be
able to output a $\hat c$ which agrees with $c$ almost everywhere with
respect to $\D$. The distribution-independent sample complexity of
learning $c$ depends fundamentally on the ``complexity'' of the set of
binary functions $\C$ from which we are choosing $\hat c$.  We define
the relevant complexity measure next.

Let $S$ be a set of $m$ samples from $\dom$. The set $S$ is said to be
\emph {shattered} by $\C$ if, for every subset $T\subseteq S$, there is
some $c_T\in\C$ such that $c_T(x) = 1$ if $x\in T$ and $c_T(y) = 0$ if
$y\notin T$.  That is, ranging over all $c \in \C$ induces all
$2^{|S|}$ possible projections onto $S$.
The {\em VC dimension} of $\C$, denoted $\VC(\C)$, is the size
of the largest set $S$ that can be shattered by $\C$. 

Let $\err_S(\hat c) = (\sum_{x\in S} |c(x) - \hat{c}(x)|)/|S|$ denote
the empirical error of $\hat c$ on $S$, and let
$\err(\hat c) = \E_{x\sim D}[|c(x) - \hat{c}(x)|]$ denote the
true expected error of $\hat c$ with respect to $\D$.  A key
result from learning theory~\cite{VC} is: for every distribution $\D$,
a sample $S$ of size
$\Omega (\eps^{-2}(\VC(\C) + \ln\frac{1}{\delta}))$ is sufficient to
guarantee that
$\err_{S}(\hat c) \in [\err(\hat c) - \epsilon, \err(\hat c) +
\epsilon]$
for \emph{every} $\hat c\in \C$ with probability $1-\delta$.  In this
case, the error on the sample is close to the true error,
simultaneously for every hypothesis in $\C$.  In particular, choosing
the hypothesis with the minimum sample error minimizes the true error,
up to $2\eps$. We say $\C$ is {\em $(\epsilon, \delta)$-uniformly
  learnable with sample complexity $m$} if, given a sample $S$ of size
$m$, with probability $1-\delta$, for all $c\in \C$,
$|\err_S(c) - \err(c)| < \epsilon$: thus, any class $\C$ is
$(\eps, \delta)$-uniformly learnable with
$m = \Theta\left(\frac{1}{\eps^2}\left(\VC(\C) + \ln
    \frac{1}{\delta}\right)\right)$
samples.  Conversely, for every learning algorithm $\A$ that uses
fewer than $\frac{\VC(\C)}{\epsilon}$ samples, there exists a
distribution $\D'$ and a constant $q$ such that, with probability at
least $q$, $\A$ outputs a hypothesis $\hat c'\in \C$ with
$\err(\hat c') > \err(\hat c) + \frac{\epsilon}{2}$ for some
$\hat c \in \C$.  That is, the true error of the output hypothesis is
more than $\frac{\epsilon}{2}$ larger the best hypothesis in the
class.

To learn real-valued functions, we need a generalization of VC
dimension (which concerns binary functions). The
\emph{pseudo-dimension} \cite{pollard1984} does exactly
this. Formally, let $c : \dom \to [0,H]$ be a real-valued function
over $\dom$, and $\C$ the class we are learning over. Let $S$ be a
sample drawn from $\D$, $|S|=m$, labeled according to $c$.  Both the
empirical and true error of a hypothesis $\hat c$ are defined as
before, though $|\hat c(x) - c(x)|$ can now take on values in
$[0, H]$ rather than in $\{0,1\}$. Let
$(r^1, \ldots, r^m) \in [0,H]^m$ be a set of \emph{targets} for
$S$. We say $(r^1, \ldots, r^m) $ \emph{witnesses} the shattering of
$S$ by $\C$ if, for each $T\subseteq S$, there exists some $c_T\in \C$
such that $f_T(x^i) \geq r^i$ for all $x^i \in T$ and $c_T(x^i) < r^i$
for all $x^i \notin T$. If there exists some $\vec r$ witnessing the
shattering of $S$, we say $S$ is {\em shatterable} by $\C$.  The {\em
  pseudo-dimension} of $\C$, denoted $\pd_\C$, is the size of the
largest set $S$ which is shatterable by $\C$.  The sample complexity
upper bounds of this paper are derived from the following theorem,
which states that the distribution-independent sample complexity of
learning over a class of real-valued functions $\C$ is governed by the
class's pseudo-dimension.

\begin{theorem}\label{thm:fat-sample}[E.g.~\cite{AB}]
  Suppose $\C$ is a class of real-valued functions with range in $[0,H]$ and
  pseudo-dimension $\pd_\C$. For every $\epsilon > 0, \delta \in [0,1]$, the
  sample complexity of $(\epsilon, \delta)$-uniformly learning $f$
  with respect to $\C$ is $ m = O\left(
    \left(\frac{H}{\epsilon}\right)^2\left(\pd_\C\ln\left(\frac{H}{\epsilon}\right)
      + \ln\left(\frac{1}{\delta}\right)\right)\right).$
\end{theorem}
Moreover, the guarantee in Theorem~\ref{thm:fat-sample}
is realized by the learning algorithm that simply outputs the
function $c \in \C$ with the smallest empirical error on the sample.

\paragraph{Applying Pseudo-Dimension to Auction
  Classes}\label{sec:applying}

For the remainder of this paper, we consider classes of truthful
auctions $\C$.\footnote{An auction is {\em truthful} if truthful
  bidding is a dominant strategy for every bidder.  That is: for every
  bidder~$i$, and all possible bids by the other bidders, $i$
  maximizes its expected utility (value minus price paid) by bidding
  its true value.  In the single-parameter settings that we study,
the expected revenue of the optimal non-truthful auction (measured at
a Bayes-Nash equilibrium with respect to the prior distribution) is no
larger than that of the optimal truthful auction.}
When we discuss some auction $c\in \C$, we treat
$c : [0,H]^n \to \mathbb{R}$ as the function that maps (truthful) bid
tuples to the revenue achieved on them by the auction $c$.  Then,
rather than minimizing error, we aim to maximize revenue.  In our
setting, the guarantee of Theorem~\ref{thm:fat-sample} directly
implies that, with probability at least $1-\delta$ (over the $m$
samples), the output of the {\em empirical revenue maximization}
learning algorithm --- which returns the auction $c \in \C$ with the
highest average revenue on the samples --- chooses an auction with
expected revenue (over the true underlying distribution $F$) that is
within an additive $\eps$ of the maximum possible. 

%In the following
%section, we introduce a parameterized class of auctions, which we name
%$\t$-level auctions. We show that $\t$-level single-item auctions have
%$O\left(n \t \log (n\t)\right)$ pseudo-dimension
%(Theorem~\ref{thm:fat-single}). Furthermore, for polynomial $\t$, that
%there exists a $\t$-level auction whose revenue is a
%$1-\eps$-approximation to the optimal revenue, when bidders' valuation
%distributions are bounded (Theorem~\ref{thm:rep-bounded}).  We extend
%these results to unbounded, MHR bidder valuations, to matroid
%settings, and to general single-parameter feasibility settings in the
%full version of the paper, all of which are substantial
%generalizations of the single-item bounded case.

%TODO
%-expand significantly captions to Figs 1 and 2

\section{Single-Item Auctions}\label{s:si}

To illustrate out ideas, we
first focus on single-item auctions.
%, a setting in which our key
%definitions and proof techniques are easily understood; t
The results of this section are generalized significantly in
Sections~\ref{sec:matroid} and~\ref{sec:single-param}.
%full version of the
%paper. 
Section~\ref{sec:t-level} defines the class of $t$-level
single-item auctions, gives an example, and interprets the auctions as
approximations to virtual welfare
maximizers. Section~\ref{sec:fat-shattering} proves that the
pseudo-dimension of the set of such auctions is $O(nt \log nt)$, which
by Theorem~\ref{thm:fat-sample} implies a sample-complexity upper
bound.  Section~\ref{sec:representation} proves that taking
$t = \Omega(\tfrac{H}{\eps})$ yields low representation error.

\subsection{$t$-Level Auctions: The Single-Item Case}\label{sec:t-level}

We now introduce {\em $t$-level auctions}, or $\Ct$ for
short. Intuitively, one can think of each bidder as facing one of $t$
possible prices; the price they face depends upon the values of the
other bidders. Consider, for each bidder $i$, $t$ numbers
$0\leq \lev{i}{0} \leq\lev{i}{1}\leq\ldots \leq \lev{i}{t-1}$. We
refer to these $t$ numbers as \emph{thresholds}. This set of $tn$
numbers defines a $t$-level auction with the following allocation
rule. Consider a valuation tuple $\v$:
\begin{enumerate}[noitemsep,nolistsep]

\item For each bidder~$i$, let $t_i(v_i)$ denote the 
index $\ind$ of the largest threshold
$\lev{i}{\ind}$ that lower bounds $v_i$ (or -1 if $v_i <
  \lev{i}{0}$). 
We call $t_i(v_i)$ the {\em level} of bidder~$i$.
\item Sort the bidders from highest level to lowest level and, within
  a level, use a fixed lexicographical tie-breaking ordering $\succ$
  to pick the winner.\footnote{Our results hold also for some other
    tie-breaking rules.}
%When bidder's valuation distributions
%    are regular, this tie-breaking can be done by value, or randomly;
%    when it is done by value, this equates to a generalization of \vcg
%    with nonanonymous reserves (and is IC and has identical
%    representation error as this analysis when bidders are regular).},
\item Award the item to the first bidder in this sorted order (unless $t_i
  = -1$ for every bidder~$i$, in which case there is no sale).
\end{enumerate}

The payment rule is the unique one that renders truthful bidding a
dominant strategy and charges~0 to losing bidders --- that is, the winning
bidder pays the lowest bid at which she would continue to win.  It is
important for us to understand this payment rule in detail; there are
three interesting cases.  Suppose bidder $i$ is the winner.  In the
first case, $i$ is the only bidder who might be allocated the item
(other bidders have level -1), in which case her bid must be at least
her lowest threshold.  In the second case, there are multiple bidders
at her level, so she must bid high enough to be at her level (and,
since ties are broken lexicographically, this is her threshold to
win).  In the final case, she need not compete at her level: she can
choose to either pay one level above her competition (in which case
her position in the tie-breaking ordering does not matter) or she can
bid at the same level as her highest-level competitors (in which case
she only wins if she dominates all of those bidders at the
next-highest level according to $\succ$).  Formally, the payment $p$
of the winner $i$ (if any) is as follows.  Let $\bar\ind$ denote the
highest level $\ind$ such that there at least two bidders at or above
level $\ind$, and $I$ be the set of bidders other than~$i$ whose level
is at least $\bar\ind$.

\begin{itemize}[noitemsep,nolistsep]
\item[Monop] If $\bar\ind = -1$, then $p_i = \lev{i}{0}$ (she is the
  only potential winner, but must have level $\geq 0$ to win).
\item[Mult] If $t_i(v_i) = \bar\ind$ then $p_i = \lev{i}{\bar\ind}$
  (she needs to be at level $\bar\ind$).
\item[Unique] If $t_i(v_i) > \bar \ind$, if $i \succ i'$ for all
  $i'\in I$, she pays $p_i = \lev{i}{\bar\ind}$, otherwise she pays
  $p_i = \lev{i}{\bar \ind + 1}$ (she either needs to be at level
  $\bar \ind + 1$, in which case her position in $\succ$ does not
  matter, or at level $\bar\ind$, in which case she would need to be
  the highest according to $\succ$).
\end{itemize}

We now describe a particular $t$-level auction, and demonstrate each
case of the payment rule.
\begin{example}\label{ex:levels-payment}
  Consider the following $4$-level auction for bidders $a,b,c$. Let
  $\lev{a}{\cdot}= [2,4,6,8]$, $\lev{b}{\cdot} = [1.5, 5, 9, 10]$, and
  $\lev{c}{\cdot} = [1.7, 3.9, 6, 7]$.  For example, if bidder $a$
  bids less than $2$ she is at level $-1$, a bid in $[2,4)$ puts her
  at level $0$, a bid in $[4,6)$ at level $1$, a bid in $[6, 8)$ at
  level $2$, and a bid of at least $8$ at level $3$.  Let
  $a \succ b \succ c$.
\begin{itemize}[noitemsep,nolistsep]
\item[Monop] If $v_a = 3, v_b < 1.5, v_c < 1.7$, then $b, c$ are at level $-1$ (to
  which the item is never allocated). So, $a$ wins and pays $2$, the
  minimum she needs to bid to be at level $0$.
\item[Mult] If $v_a \geq 8, v_b \geq 10, v_c < 7$, then $a$ and $b$
  are both at level $3$, and $a \succ b$, so $a$ will win and pays $8$
  (the minimum she needs to bid to be at level $3$).
\item[Unique] If $v_a \geq 8, v_b \in [5,9], v_c \in [3.9,6]$, then
  $a$ is at level $3$, and $b$ and $c$ are at level $1$. Since
  $a\succ b$ and $a\succ c$, $a$ need only pay $4$ (enough to be at
  level $1$).  If, on the other hand, $v_a \in [4,6], v_b = [5,9]$ and
  $v_c \geq 6$, $c$ has level at least $2$ (while $a,b$ have level
  $1$), but $c$ needs to pay $6$ since $a, b \succ c$.
\end{itemize}
\end{example}

\begin{remark}[Connection to virtual valuation functions]
  $t$-level auctions are naturally interpreted as discrete
  approximations to virtual welfare maximizers, and our representation
  error bound in Theorem~\ref{thm:rep-bounded} makes this
  precise. Each level corresponds to a constraint of the form ``If any
  bidder has level at least $\ind$, do not sell to any bidder with
  level less than $\ind$.''  We can interpret the $\lev{i}{\ind}$'s
  (with fixed $\ind$, ranging over bidders~$i$) as the bidder values
  that map to some common virtual value.  For example, $1$-level
  auctions treat all values below the single threshold as having
  negative virtual value, and above the threshold uses values as
  proxies for virtual values.  $2$-level auctions use the second
  threshold to the refine virtual value estimates, and so on.  With
  this interpretation, it is intuitively clear that as $t\to\infty$,
  it is possible to estimate bidders' virtual valuation functions and
  thus approximate the optimal auction to arbitrary accuracy.
\end{remark}
\subsection{The Pseudo-Dimension of $t$-Level
  Auctions}\label{sec:fat-shattering}

This section shows that the pseudo-dimension of the class of
$t$-level single-item auctions with $n$ bidders is
$O(nt \log nt)$.  Combining this with Theorem~\ref{thm:fat-sample}
immediately yields sample complexity bounds (parameterized by $t$) for
learning the best such auction from samples.

\begin{theorem}\label{thm:fat-single}
For a fixed tie-breaking order $\succ$,
the pseudo-dimension of the set of $n$-bidder single-item $t$-level
  auctions is $O\left(nt \log(nt)\right)$.
\end{theorem}

\begin{proof}
  Recall from Section~\ref{sec:learning} that we need to upper bound
  the size of every set that is shatterable using $t$-level auctions.
  Fix a set of samples
  $S = \left(\mathbf{v^1}, \ldots, \mathbf{v^m}\right)$ of size $m$
  and a potential witness $R = \left(r^1, \ldots, r^m\right)$. Each
  auction $c$ induces a binary labeling of the samples $\vals^j$ of
  $S$ (whether $c$'s revenue on $\vals^j$ is at least $r^j$ or
  strictly less than $r^j$).  The set $S$ is shattered with witness
  $R$ if and only if the number of distinct labelings of $S$ given by
  any $t$-level auction is $2^m$.

  We upper-bound the number of distinct labelings of $S$ given by
  $t$-level auctions (for some fixed potential witness~$R$), counting
  the labelings in two stages.  Note that $S$ involves $nm$ numbers
  --- one value $\val_i^j$ for each bidder for each sample.  A
  $t$-level auction involves $nt$ numbers --- $t$ thresholds
  $\ell_{i,\tau}$ for each bidder.  Call two $t$-level auctions with
  thresholds $\{ \ell_{i,\tau} \}$ and $\{ \hat \ell_{i,\tau} \}$ {\em
    equivalent} if:
\begin{enumerate}[noitemsep,nolistsep]

\item The relative order of the $\ell_{i,\tau}$'s agrees with that of
  the $\hat \ell_{i,\tau}$'s, in that both induce the same permutation
  of $\{1,2,\ldots,n\} \times \{0,1,\ldots,t-1\}$.

\item Merging the sorted list of the $\val_{i}^j$'s with the sorted
  list of the $\ell_{i,\tau}$'s yields the same partition of the 
$\val_i^j$'s as does merging it with the sorted list of the $\hat
  \ell_{i,\tau}$'s.

\end{enumerate}

Note that this is an equivalence relation.  If two $\t$-level auctions
are equivalent, every comparison between two numbers (valuations or
thresholds) is resolved identically by those auctions. Using the
defining properties of equivalence, a crude upper bound on the number
of equivalence classes is
\begin{align}\label{eq:classes}
(nt)! \cdot \binom{nm+nt}{nt} \le (nm+nt)^{2nt}.
\end{align}
We now upper-bound the number of distinct labelings of $S$ that can be
generated by any auction in a single equivalence class $C$.  First, as
all comparisons between two numbers (valuations or thresholds) are
resolved identically for all auctions in $C$, each bidder $i$ in each
sample $\vals^j$ of $S$ is assigned the same level (across auctions in
$C$), and the winner (if any) in each sample $\vals^j$ is constant
across all of $C$.  By the same reasoning, the identity of the
parameter that gives the winner's payment (some $\ell_{i,\tau}$) is
uniquely determined by pairwise comparisons (recall
Section~\ref{sec:t-level}) and hence is common across all auctions in
$C$. The payments $\ell_{i,\tau}$, however, can vary across auctions
in the equivalence class.

For a bidder $i$ and level $\tau \in \{0,1,2,\ldots,t-1\}$, let
$S_{i,\tau} \sse S$ be the subset of samples in which bidder~$i$ wins
and pays $\ell_{i,\tau}$.  The revenue obtained by each auction in $C$
on a sample of $S_{i,\tau}$ is simply $\ell_{i,\tau}$ (and independent
of all other parameters of the auction).  Thus, ranging over all
$t$-level auctions in $C$ generates at most $|S_{i,\tau}|$ distinct
binary labelings of $S_{i,\tau}$ --- the possible subsets of
$S_{i,\tau}$ for which an auction meets the corresponding target $r^j$
form a nested collection.

Summarizing, within the equivalence class $C$ of $t$-level auctions,
varying a parameter $\ell_{i,\tau}$ generates at most $|S_{i,\tau}|$
different labelings of the samples $S_{i,\tau}$ and has no effect on
the other samples.  Since the subsets $\{ S_{i,\tau} \}_{i,\tau}$ are
disjoint, varying all of the $\ell_{i,\tau}$'s (i.e., ranging over
$C$) generates at most
\begin{align}\label{eq:labelings}
\prod_{i=1}^n \prod_{\tau=0}^{t-1} |S_{i,\tau}| \le m^{nt}
\end{align}
distinct labelings of $S$.

Combining~\eqref{eq:classes} and~\eqref{eq:labelings}, the class of
all $t$-level auctions produces at most $(nm+nt)^{3nt}$ distinct
labelings of $S$.  Since shattering $S$ requires $2^m$ distinct
labelings, we conclude that $2^m \le (nm+nt)^{3nt},$ implying
$m = O(nt \log nt)$ as claimed.
\end{proof}

\subsection{The Representation Error of Single-Item $\t $-Level 
  Auctions}\label{sec:representation}

In this section, we show that for every bounded product distribution,
there exists a $\t $-level auction with expected revenue close to that
of the optimal single-item auction when bidders are independent and
bounded. The analysis ``rounds'' an optimal auction to a $t$-level
auction without losing much expected revenue.  This is done using
thresholds to approximate each bidder's virtual value: the lowest
threshold at the bidder's monopoly reserve price, the next
$\frac{1}{\epsilon}$ thresholds at the values at which bidder $i$'s
virtual value surpasses multiples of $\epsilon$, and the remaining
thresholds at those values where bidder $i$'s virtual value reaches
powers of $1+\epsilon$.  Theorem~\ref{thm:rep-bounded} formalizes this
intuition.

\begin{theorem}\label{thm:rep-bounded}
  Suppose $F$ is product distribution over~$[1,H]^n$.  If
  $\t = \Omega\left(\frac{1}{\epsilon} + \log_{1+\epsilon} H\right)$,
  then $\Ct$ contains a single-item auction with expected revenue at least
  $1-\epsilon$ times the optimal expected revenue.
\end{theorem}

With an eye toward our generalizations, we prove the following more
general result.  Theorem~\ref{thm:rep-bounded} follows immediately by
taking $\alpha=\gamma=1$.

\begin{lemma}\label{lem:levels}
  Consider $n$ bidders with valuations in $[0,H]$ and with
  $\pr[\max_{i} v_i > \alpha] \geq \gamma$. Then, $\Ct$ contains a
  single-item auction with expected revenue at least a $1-\epsilon$
  times that of an optimal auction, for
  $t = \Theta\left(\frac{1}{\gamma\epsilon} +
    \log_{1+\epsilon}\frac{H}{\alpha} \right)$.
\end{lemma}

\begin{proof}
  Consider a fixed bidder $i$. We define $\t $ thresholds for $i$,
  bucketing $i$ by her virtual value, and prove that the $\t$-level
  auction $\A$ using these thresholds for each bidder closely
  approximates the expected revenue of the optimal auction~$\M$. Let
  $\epsilon'$ be a parameter defined later.

  Set $\lev{i}{0} = \vvi{i}{0}$, bidder $i$'s monopoly
  reserve.\footnote{Recall from Section~\ref{sec:bayesian} that
    $\phi_i$ denotes the virtual valuation function of bidder $i$.
For the non-regular case, $\phi_i$ denotes the ironed virtual
valuation functions.
%    (From here on, we always mean the ironed version of virtual
%    values.)  
It is convenient to assume that these functions are
    strictly increasing (not just nondecreasing); this can be enforced
    at the cost of losing an arbitrarily small amount of revenue.}
  For
  $\ind \in \lbrack 1, \lceil \frac{1}{\gamma \epsilon'}\rceil
  \rbrack$,
  let $\lev{i}{\ind} = \vvi{i}{\ind \cdot \alpha \gamma \epsilon'}$
  ($\phi_i \in [0,1]$).  For
  $\ind\in \lbrack\lceil \frac{1}{\gamma \epsilon'}\rceil, \lceil
  \frac{1}{\gamma \epsilon'}\rceil + \lceil \log_{1 +
    \frac{\epsilon}{2}} \frac{H}{\alpha} \rceil \rbrack$,
  let
  $\lev{i}{\ind} = \vvi{i}{\alpha(1+ \tfrac{\epsilon}{2})^{\ind -
      \lceil \frac{1}{\gamma \epsilon'}\rceil} }$ ($\phi_i > 1$).

  Consider a fixed valuation profile $\v$. Let $\winner$ denote the
  winner according to $\A$, and $\mwinner$ the winner according to the
  optimal auction $\M$. If there is no winner, we interpret
  $\vv{\winner}{v_{\winner}}$ and $\vv{\mwinner}{v_{\mwinner}}$ as 0.
  Recall that $\M$ always awards the item to a bidder with the highest
  positive virtual value (or no one, if no such bidders exist).  The
  definition of the thresholds immediately implies the following.

\begin{enumerate}[noitemsep,nolistsep]
\item $\A$ only allocates to non-negative ironed virtual-valued
  bidders. \label{nnvv}
\item If there is no tie (that is, there is a unique bidder at the
  highest level), then $\mwinner  = \winner$.\label{agree} 
\item When there is a tie at level $\ind$, the virtual value of the
  winner of $\A$ is close to that of $\M$:

 If $\ind \in \lbrack 0, \lceil \frac{1}{\gamma\epsilon'}\rceil\rbrack$ then
  $\vv{\mwinner}{v_{\mwinner}} -\vv{\winner}{v_{\winner}} \leq
\alpha \gamma \epsilon' $;  

if
$\ind\in \lbrack\lceil \frac{1}{\gamma\epsilon'}\rceil , \lceil
\frac{1}{\gamma\epsilon'}\rceil + \lceil \log_{1 + \frac{\epsilon}{2}}
\frac{H}{\alpha} \rceil \rbrack$,
$\frac{\vv{\winner}{v_{\winner}}}{\vv{\mwinner}{v_{\mwinner}}} \geq 1
- \frac{\epsilon}{2}$.
\end{enumerate}

These facts imply that
\begin{equation}\label{eq:rep1}
\E_{\v}\lbrack\rev(\A)\rbrack = 
\E_{\v}\lbrack\vv{\winner}{v_\winner}\rbrack 
\geq 
(1-\tfrac{\epsilon}{2})
\cdot 
\E_{\v}\lbrack\vv{\mwinner}{v_{\mwinner}}\rbrack - \alpha\gamma\epsilon'
=
(1-\tfrac{\epsilon}{2})
\cdot 
\E_{\v}\lbrack\rev(\M)\rbrack - \alpha\gamma\epsilon'.
\end{equation}
where the first and final equality follow from $\A$ and $\M$'s
allocations depending on ironed virtual values, not on the values
themselves, thus, the ironed virtual values are equal in expectation
to the unironed virtual values, thus the revenue, of the mechanisms
(see~\cite{hartline2015}, Chapter 3.5 for discussion).

The assumption that $\pr[\max_i v_i > \alpha] \geq \gamma$ implies
(since a feasible auction prices the good at $\alpha$ and awards it to
any bidder with value at least $\alpha$, if any) that
$ \E\lbrack \rev(\M) \rbrack \geq \alpha\gamma.$ Combining this with
~\eqref{eq:rep1}, and setting $\epsilon' = \frac{\eps}{2}$ implies
$ \E_{\v}\lbrack\rev(\A)\rbrack \ge
\left(1-\epsilon\right)\E_{\v}\lbrack\rev(\M)\rbrack$.
\end{proof}

Combining Theorems~\ref{thm:fat-sample} and~\ref{thm:rep-bounded}
yields the following Corollary~\ref{cor:sample-single}.

\begin{corollary}\label{cor:sample-single}
  Let $F$ be a product distribution with all bidders' valuations in
  $[1,H]$.  Assume that $\t = \Theta\left(\bounded\right)$ and
  $m =
  O\left(\left(\frac{H}{\epsilon}\right)^2\left(n\t\log\left(n\t\right)
      \log\frac{H}{\epsilon} + \log \frac{1}{\delta}\right)\right) =
  \tilde{O}\left(\frac{H^2n}{\epsilon^3}\right).$
  Then with probability at least $1-\delta$, the single-item empirical
  revenue maximizer of $\Ct$ on a set of $m$ samples from $F$ has
  expected revenue at least $1-\epsilon$ times that of the optimal
  auction.
\end{corollary}

\section{Unbounded MHR Distributions}\label{sec:mhr}

This section shows how to replace the assumption of bounded valuations
by the assumption that each valuation distribution satisfies the
monotone hazard rate (MHR) condition, meaning that
$\tfrac{f_i(v_i)}{1-F_i(v_i)}$ is nondecreasing.  Our resulting sample
complexity bounds depend on the number of bidders~$n$ and the error
parameter~$\eps$ only.  bounded case, following ideas from This
extension is based on previous work~\cite{cai2011extreme} that
effectively reduces the case of MHR valuations to the case of
valuations lying in the interval
$\left[\beta\epsilon, 2\beta\log{\frac{1}{\epsilon}}\right]$ for a
suitable choice of $\beta$. Our analysis works with
$\eta$-\emph{truncated} $\t-$level auctions, where each $\t$-level
auction $f$ is replaced with $f_\eta = \min(f, \eta)$.

\begin{theorem}\label{thm:rep-mhr}
  Suppose $\F$ is a product distribution and each bidder's valuation
  distribution satisfies the MHR condition. Then, for each $\eps > 0$, and each $\hat{\beta} \geq \beta$ such that 
  $\pr\left[\max_i \v_i \geq \frac{\hat{\beta}}{2}\right]\geq 1 -
  \frac{1}{\sqrt{e}} - \eps'$,
  there is a $\t $-level
  $\left(\hat{\beta}\log\frac{1}{\epsilon'}\right)$-truncated auction
  with expected revenue at least $1-\eps$ times that of an optimal
  auction, where
% which $1-\epsilon$-approximates Myerson's
%  revenue for the single-item auction, for
  $\t = \Theta\left(\frac{1}{\epsilon'} +
    \log_{1+\epsilon'}\left(\log\frac{1}{\epsilon'}\right)\right)$
  and $\epsilon' = O\left(\frac{\epsilon}{\log\frac{1}{\epsilon}}\right)$.
\end{theorem}
%
%  Assume the auction is given
%  an estimate of $\beta$, $\hat \beta$, such that
%  $\hat \beta \geq \beta$ but that
%  $\pr\lbrack \max_i v_i \geq \frac{\hat \beta}{2} \rbrack \geq
%  1-\frac{1}{\sqrt{e}}$.
%  \footnote{Since $\beta$ has quantile at least
%    $\frac{1-\frac{1}{\sqrt{e}}}{2}$, using the largest $B$ such that
%    the empirical probability of some bidder bidding above $B$ is at
%    least $\frac{1-\frac{1}{\sqrt{e}}}{2} - \epsilon$, where there
%    might be $\epsilon$ error estimating this probability if the
%    sample size is $\Omega\left(\frac{1}{\epsilon^2}\right)$, will
%    suffice. We don't actually need an accurate estimate of $\beta$,
%    just of a big enough bid with large probability of the maximum bid
%    surpassing it.}  For each sample $\v$, set
%  $ \v_i = \min(\v_i, 2\hat{\beta}\log\tfrac{1}{\epsilon'})$.

Before proving Theorem~\ref{thm:rep-mhr}, we quote a key fact 
about MHR distributions~\citet{cai2011extreme}.

\begin{theorem}[Theorem~19 and Lemma~38
  of~\cite{cai2011extreme}]\label{thm:anchoring}
  Let $X_1, \ldots, X_n$ be a collection of independent random
  variables whose distributions satisfy the MHR condition. Then there
  exists an 
  anchoring point $\beta$ such that
  $$\pr[\max_i X_i \geq \frac{\beta}{2}]\geq 1 - \frac{1}{\sqrt{e}},$$
  and for all $\epsilon > 0$,
  \[\int_{2\beta\log\frac{1}{\epsilon}}^{\infty} z f_{\max_i X_i}(z)dz
  \leq 36\beta\epsilon\log\frac{1}{\epsilon}.\]
\end{theorem}

Now, we proceed to prove Theorem~\ref{thm:rep-mhr}.

\vspace{.1in}
\noindent
\begin{prevproof}{Theorem}{thm:rep-mhr}
  Fix $\epsilon'$, to be defined later.  Conditioning on all bids
  being at most $2\hat{\beta}\log\frac{1}{\epsilon'}$ allows us apply
  Lemma~\ref{lem:levels} as though the valuations are bounded. In
  particular, since
  $\pr[\max_i v_i \geq \frac{\hat{\beta}}{2}]\geq 1 -
  \frac{1}{\sqrt{e}} - \eps'$,
  Lemma~\ref{lem:levels} implies for $\alpha = \frac{\hat\beta}{2}$,
  $\gamma = 1 - \frac{1}{\sqrt{e}}- \eps'$ and
  $H = 2\hat{\beta}\log\frac{1}{\epsilon'}$, implies the existence of
  a
  $\t = O\left(\frac{1}{\eps'}+
    \log_{1+\eps'}\left(\log\frac{1}{\eps'}\right)\right)$-level
  truncated\footnote{Lemma~\ref{lem:levels} only implies the existence
    of such a $\t$-level auction. However, when the bids are all below
    some $\eta$, one can always find an $\eta$-truncated auction which
    is equivalent to each untruncated auction.} auction $\A$ such
  that:
  \begin{align}
    \E\left[\rev(\A)|\max_i \v_i  \leq 2\hat{\beta}\log\frac{1}{\epsilon'}\right]
&\geq (1-\epsilon') \E\left[\rev(\M)|\max_i \v_i \leq 2\hat{\beta}\log\frac{1}{\epsilon'}\right]\label{eqn:cond}
\end{align}

Thus, we have
\begin{align*}
\E\left[\rev(\A)\right] &\geq \E\left[\rev(\A)|\max_i \v_i \leq 2\hat{\beta}\log\frac{1}{\epsilon'}\right] \pr\left[\max_i \v_i \leq 2\hat{\beta}\log\frac{1}{\epsilon'}\right]\\
&\geq (1-\epsilon') \E\left[\rev(\M)|\max_i \v_i\leq
   2\hat{\beta}\log\frac{1}{\epsilon'}\right]\pr\left[\max_i \v_i \leq 2\hat{\beta}\log\frac{1}{\epsilon'}\right] \\
&\geq (1-\epsilon') \E\left[\rev(\M)\right] - 36\hat{\beta}\epsilon' \log\frac{1}{\epsilon'}\\
& \geq \left(1-O\left(\epsilon'\log\frac{1}{\epsilon'}\right)\right) \E\left[\rev(\M)\right]
\end{align*}

where the first inequality comes from the fact that $\A$ only sells to
agents with non-negative virtual value (so the revenue on a smaller
region of bids is only less), the second from Equation~\ref{eqn:cond}
and probabilities being at most $1$, the penultimate from
Theorem~\ref{thm:anchoring}, and the final from the fact that
$\rev(\M) \geq \frac{1-\frac{1}{\sqrt{e}}}{2}\hat{\beta}$. Setting
$\epsilon = \epsilon'\log\frac{1}{\epsilon'}$ and noticing that this
implies $\epsilon' \leq \frac{\epsilon}{\log\frac{1}{\epsilon}}$
yields the desired result.
\end{prevproof}

Corollary~\ref{cor:mhr-sample} follows as a corollary of
Theorems~\ref{thm:fat-sample} and~\ref{thm:rep-mhr}.

\begin{corollary}\label{cor:mhr-sample}
  With probability $1-\delta$, the empirical revenue maximizer for a
  sample of size $m$ of the class of $\t $-level $\eta$-truncated
  single-item auctions is a $1-O\left(\epsilon\right)$-approximation
  of the optimal auction for $n$ MHR bidders, for
  $\t = O\left(\frac{1}{\epsilon'} +
    \log_{1+\epsilon'}\left(\log\frac{1}{\epsilon'}\right)\right)$,
  $\epsilon' = \frac{\epsilon}{\log\frac{1}{\epsilon}}$ and
  \[m = O\left(\left(\frac{1}{\eps'}\right)^2\left(nt\log\left(nt\right)\ln\frac{1}{\eps'} + \ln\frac{1}{\delta}\right)\right) = \tilde{O}\left(\frac{n}{\eps^3}\right)\]
where $\eta$ can be learned from the set of $m$ samples.
% for $\eta = \frac{\hat{\beta}}{2}\log\frac{1}{\eps'}$,
% $\hat\beta\geq \beta$ and
% $\pr\left[\max_i \v_i \geq \frac{\hat{\beta}}{2}\right]\geq 1 -
% \frac{1}{\sqrt{e}} - \jmcomment{epswhat}$.
\end{corollary}

\begin{proof}
  We first argue that one can learn some $\eta$ from the sample.  Let
  $Q(S, g) = \frac{\sum_{\v\in S_{\epsilon'}} \mathbb{I}\left[\max_i
      \v_i \geq g\right]}{|S|}$
  and $q(g) = \pr\left[\max_i \v_i \geq g\right]$ (the empirical and
  true probability that the maximum bid is at least $g$,
  respectively).  Given $\epsilon'$, consider a set of samples
  $S_{\epsilon'}$ of profiles, and compute the largest $\hat\beta$
  such that
  $Q(S_{\epsilon'}, \frac{\hat\beta}{2}) \geq 1 - \frac{1}{\sqrt{e}} -
  \eps'$.
  Standard VC-bounds imply that 
  $|q(\rho) - Q(S_{\epsilon'}, \rho)| \leq \eps'$ for all $\rho$ with
  probability at least 
  $1-\delta$, provided
  $|S_{\epsilon'}| \geq
  \Theta\left(\left(\frac{1}{\eps'}\right)^2\ln\frac{1}{\delta}\right)$.\footnote{This
    can be thought of as a class of binary classifiers
    with VC-dimension one.}.  In particular, with probability
  $1-\delta$, we will have
  $q(S_{\epsilon'}, \frac{\beta}{2}) \geq 1 - \frac{1}{\sqrt{e}}
  -\eps'$,
  so it will be the case that $\hat\beta \geq \beta$, and also
  $q(\hat\beta) \geq 1 - \frac{1}{\sqrt{e}} -2\eps'$. Then, let
  $\eta = 2\hat{\beta}\log\frac{1}{\eps'}$.

  Now, Theorem~\ref{thm:rep-mhr} implies the existence of a
  $\eta$-truncated $\t$-level auction which $(1-\eps')$-approximates
  the optimal auction. The argument is completed using the fact that, if the
  auctions' values are upper-bounded by $\eta$, one can equivalently
  think of the \emph{values} being upper-bounded by $\eta$, so
  Theorem~\ref{thm:fat-sample} implies the sample complexity bound
  allowing additive error
  $\epsilon'\eta =
  \frac{\epsilon'\hat{\beta}}{2}\log\frac{1}{\epsilon'}$.
  This error is multiplicatively at most
  $\epsilon = \epsilon'\log\frac{1}{\epsilon'}$, since
  $\rev(\M) = \Omega\left(\frac{\hat{\beta}}{2}\right)$.
\end{proof}

\begin{remark}[Near-optimality of sample complexity]
Can we do better than
Theorem~\ref{thm:fat-single}?
Can we learn an approximately optimal auction from a simpler
class --- with pseudo-dimension
$\poly(\log H, \log n, \frac{1}{\epsilon})$, say --- allowing for much
smaller sample complexity than achieved here?
The answer is negative:
lower bounds in~\citet{CR14} imply that
approximate revenue maximization requires 
sample complexityy at least linear in
the number of bidder~$n$, even when bidders' valuations
independently drawn from MHR distributions. Thus, every class
of auctions that is
sufficiently expressive to guarantee expected revenue at least
$1-\eps$ times optimal
must have pseudo-dimension that grows polynomially with $n$.
\end{remark}

\section{$\t$-Level Matroid Auctions}\label{sec:matroid}

This section extends the ideas and techniques from
Section~\ref{sec:fat-shattering} to matroid environments. The
straightforward generalization of $\t$-level auctions to matroid
environments suffices: we order the bidders by level, breaking ties
within a level by some fixed linear ordering over agents $\succ$, and
greedily choose winners according to this ordering (subject to
feasibility and to bids exceeding the lowest threshold).  The next
theorem bounds the pseudo-dimension of this more general class of
auctions.

\begin{theorem}\label{thm:fat-matroid}
  The pseudo-dimension of $\t$-level matroid auctions with $n$ bidders
  is $O(nt\log\left(nt\right))$.
\end{theorem}

The proof is conceptually similar to that of
Theorem~\ref{thm:fat-single}, though we require a more general
argument.  Our proof uses a couple of standard results from learning
theory (see e.g.~\cite{KV} for details).  The first, also known as
Sauer's Lemma, states that the number of distinct projections of a set
$S$ induced by a set system with bounded VC dimension grows only
polynomially in $|S|$.

\begin{lemma}\label{l:sauer}
Let $\C$ be a set of functions from $\dom$ to $\{0,1\}$ with VC
dimension~$d$, and $S \sse \dom$.  Then
$$
\left| \left\{ S \cap \{ x \in \dom \,:\, c(x) = 1 \} \,:\, c \in
\C\right\}\right| \le |S|^d.
$$
\end{lemma}

Recall that a {\em linear separator} in $\R^d$ is defined by
coefficients $a_1,\ldots,a_d$, and assigns $x \in \R^d$ the value~1 if
$\sum_{i=1}^d a_ix_i \ge 0$ and the value~0 otherwise.

\begin{lemma}\label{l:sep}
The set of linear separators in $\R^d$ has VC dimension $d+1$.
\end{lemma}

\vspace{.1in}
\noindent
\begin{prevproof}{Theorem}{thm:fat-matroid}
  Consider a set of samples $S$ of size $m$ which can be shattered by
  $\t$-level matroid auctions with revenue targets
  $(r^1, \ldots, r^m)$. We upper-bound the number of labelings of $S$
  possible using $\t$-level auctions, which again yields an upper
  bound on $m$.

  We partition auctions into equivalence classes, identically to the
  proof of Theorem~\ref{thm:fat-single}.  Recall that, across all
  auctions in an equivalence class, all comparisons between two
  thresholds or a threshold and a bid are resolved identically.
  Recall also that the number of equivalence classes is at most
  $(nm+nt)^{2nt}$.  We now upper bound the number of distinct
  labelings any fixed equivalence class $C$ of auctions can generate.

  Consider a class $C$ of equivalent auctions.  The allocation and
  payment rules are more complicated than in the single-item case but
  still relatively simple.  In particular, whether or not a bidder
  wins depends only on the ordering of the bidders (by level) and the
  fixed tie-breaking rule $\succ$, and thus is a function only of
  comparisons between bids and thresholds.  This implies that, for
  every sample in $S$, all auctions in the class $C$ declare the same
  set of winners.  It also implies that the payment of each winning
  bidder is a fixed threshold $\lev{\winner}{\tau}$, and the identity
  of this parameter is the same across all auctions in $C$.

  Now, encode each auction $\A\in C$ and sample $\v^j$ as an
  $nt+1$-dimensional vector as follows.  Let $x^\A_{i, \tau}$ equal
  the value of $\lev{i}{\tau}$ in the auction $\A$.  Define
  $x^\A_{nt+1} = 1$ for every $\A \in C$. Define $y^j_{i, \tau} = 1$
  if bidder $i$ is a winner paying her threshold $\lev{i}{\tau}$ for
  auctions in $C$ and $0$ otherwise. Finally, define
  $y^j_{nt+1} = -r^j $.  The point is that, for every auction $\A$ in
  the class $C$ and sample $\v^j$,
  \[x^\A \cdot y^j \geq 0\]
  if and only if $\rev(\A) \geq r^j$.  
Thus, the number of distinct labelings of the samples generated by
auctions in $C$ is bounded above by the number of distinct sign
patterns on $m$ points in $\R^{nt+1}$ generated by all linear
separators.  (The $y^j$-vectors are constant across $C$ and can be
viewed as $m$ fixed points in $\R^{nt+1}$; each auction $\A \in C$
corresponds to the vector $x^{\A}$ of coefficients.)
Applying Lemmas~\ref{l:sauer} and~\ref{l:sep}, 
$\t$-level matroid auctions can generate at most $m^{nt+2}$
  labelings per equivalence class, and hence at most
  $(nm + nt)^{3nt + 2}$ distinct labelings in total.
This imposes the restriction
\[2^m \leq (nm + nt)^{3nt + 2};\]
solving for $m$ yields the desired bound.
\end{prevproof}

We now extend our representation error bound for $\t$-level
single-item auctions to matroids.

\begin{theorem}\label{thm:rep-matroid}
Consider an arbitrary matroid environment.
Suppose $F$ is a production distribution with valuations in~$[1,H]$.
  Provided
  $\t = \Omega\left(\frac{1}{\epsilon} + \log_{1+\epsilon} H\right)$,
  there exists a $\t $-level matroid auction with expected
  revenue at least a $1-\epsilon$ fraction of the optimal expected
  revenue.
\end{theorem}

The key new idea in the proof is to exhibit a bijection between the
feasible sets $\winners$ (our winning set) and $\mwinners$ (the
optimal winning set) such that each bidder from $\winners$ has a level
at least as high as their bijective partner in $\mwinners$.  To
implement this, we use the following property of matroids
(e.g.~\cite{hartline2009simple,RS15}).

\newcommand{\opt}{\textsc{Opt}}
\begin{prop}\label{prop:matroid}
  Let $\opt$ denote the largest-weight set of a matroid, and let $B$
  be any other feasible set such that $|B| = |\opt|$, and
  $\opt_i, B_i$ denote the $i$-th largest element of $\opt$ and $B$,
  respectively. Then $w(\opt_i) \geq w(B_i)$ for all $i$.
\end{prop}

\vspace{.1in}
\noindent
\begin{prevproof}{Theorem}{thm:rep-matroid}
  Define bidders' thresholds exactly as in the proof of
  Theorem~\ref{thm:rep-bounded} and let $\A$ denote the corresponding
  $\t$-level auction.  Fix an arbitrary valuation profile $\v$.  Let
  $\winners$ denote the set of winning bidders in $\A$ and $\mwinners$
  the set of winning bidders in $\M$.  Recall that the latter is the
  feasible set that maximizes the sum of virtual valuations.  Both
  sets are maximally independent amongst those bidders with
  non-negative virtual value ($\M$, by virtual of being
  welfare-maximal, and $\A$, by definition). Then, we claim
  $|\winners| = |\mwinners|$ (if not, by the augmentation property of
  matroids, the smaller set could be extended to include an element of
  the larger while maintaining independence, violating their
  maximality).

  Notice that $\winners$ is lexicographically optimal with respect to
  the \emph{levels}, rather than the exact
  weights. Proposition~\ref{prop:matroid} implies that $\mwinners$ is
  also lexicographically optimal with respect to the levels; thus, the
  level of the $i$th largest bidder in $\mwinners$ has the same level
  as the $i$th largest bidder in $\winners$.  Then, by an accounting
  argument identical to the one for the single-item case (comparing
  virtual values for the $i$th bidder in $\winners$ to the $i$th
  bidder in $\mwinners$) summing up over all bidders completes the
  proof.
\end{prevproof}

Thus, we have the following corollary about the sample complexity of
$1-\epsilon$-approximating Myerson in matroid environments with
$\t$-level auctions, noting that the maximum revenue is now $nH$
rather than $H$.

\begin{corollary}\label{cor:matroid-sample}
  With probability $1-\delta$, the empirical revenue maximizer for a
  sample of size $m$ of the class of $\t $-level single-item auctions
  is a $1-O\left(\epsilon\right)$-approximation to Myerson for $n$
  bidders whose valuations are in $[1,H]$, for $\t = O\left(\bounded\right)$ and
  \[m = O \left(\left(\frac{Hn}{\epsilon}\right)^2\left(n\t\log(n\t)\ln\frac{Hn}{\epsilon} + \ln\frac{1}{\delta}\right)\right) = \tilde{O}\left(\frac{H^2n^3}{\epsilon^3}\right).\]
\end{corollary}

 \section{Single-parameter $\t $-level
  auctions}\label{sec:single-param}

In this section, we show how to extend the ideas and techniques from
Section~\ref{sec:fat-shattering} to any single-parameter environment
which has the empty set as a feasible outcome. With this mild
assumption, the results in this section do not require the environment
to be a matroid or even downwards-closed. Before we state this result,
we need a slight generalization of the $t$-level auction to this
setting. Previously, no $\t$-level auction would allocate to any
bidder whose value was below their lowest threshold, and this will not
be a possibility in environments which are not downwards-closed.
Instead, in this setting, if \emph{any} bidder fails to pass her
lowest threshold, we will assume $\A$ will choose the empty
outcome. Moreover, a $\t$-level auction will now need more about what
the various levels represent: 
previously, we implicitly used the $k$th
threshold to correspond to a value where each bidder's virtual value
would pass some quantity $q_k$.

In this general setting, we make that connection explicit. There
will still be $nt$ numbers which define a particular $t$-level
auction, the $t$ threshold locations per bidder. In addition, we will
consider a fixed vector $\Phi\in \mathbb{R}^t$ (\emph{not}
parameterizing the auction class) which, for all $\ind$, intuitively
assigns an estimate of $\vvi{i}{\lev{i}{\ind}}$, which is the same for
all bidders $i$.  Formally, $\Phi_\ind$ will be used to assign a real
value to a feasible set $X\in \X$ with a valuation profile $\v$ as
follows.  Let $e_X =\sum_{i\in X} \Phi_{t_i(v_i)} $, where $t_i(v_i)$
as before is the level agent $i$'s bid according to $\v_i$. Then, a
particular $\t$-level auction will choose the winning set $X\in \X$
which maximizes $e_X$ (breaking ties in some fixed way which does not
depend upon the bids). If, for each bidder $i$ and level $\tau$, the
threshold $\lev{i}{\tau}$ is placed exactly at the value at which
$i$'s virtual valuation surpasses $\Phi_\tau$, then this auction is
approximately optimizing the virtual surplus of the winning set.

\begin{theorem}\label{thm:fat-param} The pseudo-dimension of
  $\t$-level single-parameter auctions with $n$ bidders is
  $O\left(n\t\log(n\t)\right)$.
\end{theorem}

These auctions have slightly more complicated payment rules than those
for matroids, where each agent $i$ was intuitively competing with (at
most) one other bidder for inclusion in the winning set. Now, a bidder
$i$ who is in the winning set $X$ will have a payment of the following
form. For a fixed assignment of levels to bidders, sort the
alternatives according to their values $e_Y$ for all $Y \in \X$. Let
$Y$ be the highest-ranked alternative set which does not contain
$i$. Then, $i$'s payment will be the threshold corresponding to the
minimal $\ind$ such that
$\sum_{i' \in X, i' \neq i} \Phi_{t_{i'}(v_{i'})} + \Phi_\ind \geq
\sum_{i'\in Y} \Phi_{t_{i'}(v_{i'})}$
(namely, the minimal bid which keeps $X$ preferred to $Y$ in terms of
the estimated virtual values)\footnote{Since we assume ties are broken
  in a way which does not depend on the bids, we can ignore ties in
  the payment rule, and agents will only ever pay thresholds.}. While
this rule is more complicated, it is \emph{still} the case that, once
each bidder is assigned to some level, each of the bidders in the
winning set's payment is just one of their thresholds. Thus, the proof
of Theorem~\ref{thm:fat-param} is identical to the one of
Theorem~\ref{thm:fat-matroid} without ties.

When considering non-downwards-closed environments, the optimal
revenue may be arbitrarily close to $0$, making it difficult to argue
about multiplicative approximations to the optimal revenue. Instead, we
will give a weaker guarantee, namely, that the empirical revenue
maximizer will have expected revenue which is \emph{additively} close
to the optimal expected revenue. If one is willing to restrict the
environment to be downwards-closed, it is possible to achieve a
multiplicative guarantee, since in that case the optimal revenue is at
least $1$.  We now state the Theorem which bounds the representation
error of $\t$-level auctions for single-parameter environments.
\begin{theorem}\label{thm:rep-param}
  There is a $\t$-level auction whose expected revenue is within an
  additive $\epsilon$ of optimal, in any single-parameter setting $\X$
  such that $\emptyset\in\X$, for
  $\t = O\left(\frac{Hn^2}{\epsilon}\right)$, for $n$ bidders with
  valuations distributions which are product and bounded in $[0,H]$.
\end{theorem}

\begin{proof}
  Let $\t = \frac{Hn^2}{\epsilon} + \frac{Hn}{\epsilon}$.  We will
  begin by defining $\Phi$, the $\t$-dimensional vector corresponding
  to the estimated virtual values.  Let $\Phi_0 = -Hn$ (if any bidder
  has virtual value $<-Hn$, the virtual value of any set containing
  her is negative, since virtual values are upper-bounded by values
  and the value of the remaining set may be at most $H(n-1)$, so in
  this case one should allocate to $\emptyset$). Then, let
  $\Phi_{\ind} = \Phi_{\ind - 1} + \frac{\epsilon}{n}$. Thus, we
  partition the space of virtual values into additive sections of
  width $\frac{\epsilon}{n}$.

  Then, for each bidder $i$ and $\ind$, let
  $\lev{i}{\ind} = \vvi{i}{\Phi_\ind}$, the value at which bidder
  $i$'s virtual value surpasses $\Phi_\ind$. Then, consider a
  valuation profile $\v$ on which $\M$ and this particular $\A$
  disagree on the winning sets $\winners, \mwinners \in \X$. Notice
  that each bidder $i$'s virtual value is estimated correctly within
  an additive $\frac{\epsilon}{n}$ by $\Phi_{t_i(v_i)}$ (assuming no
  bidder has highly negative virtual value, in which case $\M$ and
  $\A$ both choose outcome $\emptyset$), and are never
  overestimated. Thus, it is the case that

\[\sum_{\winner\in\winners} \vv{\winner}{\v_\winner} \geq \sum_{\winner\in \winners}\Phi_{t_\winner(v_\winner)} \geq \sum_{\mwinner\in \mwinners}\Phi_{ t_\mwinner(v_\mwinner)} \geq \sum_{\mwinner\in\mwinners} \vv{\mwinner}{\v_\mwinner}  - \epsilon\]

and the claim follows.
\end{proof}

Thus, we have the following sample complexity result for general
single-parameter settings.
\begin{corollary}\label{cor:sample-param}
  With probability $1-\delta$, the empirical revenue maximizer on $m$
  samples $S$ from the class of $\t$-level auctions has true expected
  revenue within an additive $\epsilon$ of Myerson's expected revenue,
  for the single-parameter environment $\X$ when bidders have
  valuations in $[0,H]$, for $\t = O\left(\frac{Hn^2}{\epsilon}\right)$, and
  \[ m = O\left(\left(\frac{Hn}{\epsilon}\right)^2\left(n\t \log(n\t)\log\frac{Hn}{\epsilon} + \log\frac{1}{\delta}\right)\right) = \tilde O\left( \frac{H^3 n^5}{\epsilon^3}\right).\]
\end{corollary}

\section*{Open Questions}

There are some significant opportunities for follow-up
research.  First, there is much to do on the design of {\em
  computationally efficient} (in addition to sample-efficient)
algorithms for learning a near-optimal auction.  The present work
focuses on sample complexity, and our learning algorithms are
generally not computationally efficient.\footnote{There is a clear
  parallel with
 {\em computational} learning theory~\cite{valiant1984theory}:
  while the information-theoretic foundations of classification (VC
  dimension, etc.~\cite{VC}) have been long understood, this 
  research area strives to understand which low-dimensional concept
  classes are learnable in polynomial time.} 
%(For example, computing the
%$t$-level auction that maximizes average revenue across a set of
%samples is, in general, an NP-hard
%problem.)  %XXX REDUCTION FROM SET COVER
%On the other hand, the specialized analyses in \cite{CR14,RS15}
%demonstrate that, at least in several interesting special cases,
%sample efficiency need not come at the expense of computational
%efficiency.
%\footnote{There is a clear parallel between this direction
%  and {\em computational} learning theory~\cite{valiant1984theory}:
%  while the information-theoretic foundations of classification (VC
%  dimension, etc.~\cite{VC}) have been long understood, this lively
%  research area strives to understand which low-dimensional concept
%  classes learnable in polynomial time.}  
The general research agenda here is
to identify auction classes $\C$ for various settings such that:

\begin{enumerate}[noitemsep,nolistsep]
\item $\C$ has low representation error;

\item $\C$ has small pseudo-dimension;

\item There is a polynomial-time algorithm to find an approximately
  revenue-maximizing auction from $\C$ on a given set of
  samples.\footnote{The sample-complexity and performance bounds
    implied by pseudo-dimension analysis, as in
    Theorem~\ref{thm:fat-sample}, hold with such an approximation
    algorithm, with the algorithm's approximation factor carrying
    through to the learning algorithm's guarantee.  See
    also~\cite{balcan2008reducing,dughmi2014sampling}.}

\end{enumerate}

There are also interesting open questions on the statistical side,
notably for multi-parameter problems.  While the negative result
in~\cite{dughmi2014sampling} rules out a universally good
upper bound on the sample complexity of learning a near-optimal
mechanism in multi-parameter settings, we suspect that positive
results are possible for several interesting special cases.

{\footnotesize{\bibliography{sources}}}

\begin{thebibliography}{24}
\providecommand{\natexlab}[1]{#1}
\providecommand{\url}[1]{\texttt{#1}}
\expandafter\ifx\csname urlstyle\endcsname\relax
  \providecommand{\doi}[1]{doi: #1}\else
  \providecommand{\doi}{doi: \begingroup \urlstyle{rm}\Url}\fi

\bibitem[Anthony and Bartlett(1999)]{AB}
Martin Anthony and Peter~L. Bartlett.
\newblock \emph{Neural Network Learning: Theoretical Foundations}.
\newblock Cambridge University Press, NY, NY, USA, 1999.

\bibitem[Babaioff et~al.(2015)Babaioff, Immorlica, Lucier, and
  Weinberg]{moshematt2014focs}
Moshe Babaioff, Nicole Immorlica, Brendan Lucier, and S.~Matthew Weinberg.
\newblock A simple and approximately optimal mechanism for an additive buyer.
\newblock \emph{SIGecom Exch.}, 13\penalty0 (2):\penalty0 31--35, January 2015.

\bibitem[Balcan et~al.(2007)Balcan, Blum, and Mansour]{balcan2007CMUtechreport}
Maria-Florina Balcan, Avrim Blum, and Yishay Mansour.
\newblock Single price mechanisms for revenue maximization in unlimited supply
  combinatorial auctions.
\newblock Technical report, Carnegie Mellon University, 2007.

\bibitem[Balcan et~al.(2008)Balcan, Blum, Hartline, and
  Mansour]{balcan2008reducing}
Maria-Florina Balcan, Avrim Blum, Jason~D Hartline, and Yishay Mansour.
\newblock Reducing mechanism design to algorithm design via machine learning.
\newblock \emph{Jour. of Comp. and System Sciences}, 74\penalty0 (8):\penalty0
  1245--1270, 2008.

\bibitem[Cai and Daskalakis(2011)]{cai2011extreme}
Yang Cai and Constantinos Daskalakis.
\newblock Extreme-value theorems for optimal multidimensional pricing.
\newblock In \emph{Foundations of Computer Science (FOCS), 2011 IEEE 52nd
  Annual Symposium on}, pages 522--531, Palm Springs, CA, USA., Oct 2011. IEEE.

\bibitem[Chawla et~al.(2007)Chawla, Hartline, and
  Kleinberg]{chawla2007algorithmic}
Shuchi Chawla, Jason Hartline, and Robert Kleinberg.
\newblock Algorithmic pricing via virtual valuations.
\newblock In \emph{Proceedings of the 8th ACM Conf. on Electronic Commerce},
  pages 243--251, NY, NY, USA, 2007. ACM.

\bibitem[Chawla et~al.(2010)Chawla, Hartline, Malec, and
  Sivan]{chawla2010multi}
Shuchi Chawla, Jason~D. Hartline, David~L. Malec, and Balasubramanian Sivan.
\newblock Multi-parameter mechanism design and sequential posted pricing.
\newblock In \emph{Proceedings of the Forty-second ACM Symposium on Theory of
  Computing}, pages 311--320, NY, NY, USA, 2010. ACM.

\bibitem[Cole and Roughgarden(2014)]{CR14}
Richard Cole and Tim Roughgarden.
\newblock The sample complexity of revenue maximization.
\newblock In \emph{Proceedings of the 46th Annual ACM Symposium on Theory of
  Computing}, pages 243--252, NY, NY, USA, 2014. SIAM.

\bibitem[Devanur et~al.(2011)Devanur, Hartline, Karlin, and
  Nguyen]{devanur2011prior}
Nikhil Devanur, Jason Hartline, Anna Karlin, and Thach Nguyen.
\newblock Prior-independent multi-parameter mechanism design.
\newblock In \emph{Internet and Network Economics}, pages 122--133. Springer,
  Singapore, 2011.

\bibitem[Dhangwatnotai et~al.(2010)Dhangwatnotai, Roughgarden, and
  Yan]{dhangwatnotai2010revenue}
Peerapong Dhangwatnotai, Tim Roughgarden, and Qiqi Yan.
\newblock Revenue maximization with a single sample.
\newblock In \emph{Proceedings of the 11th ACM Conf. on Electronic Commerce},
  pages 129--138, NY, NY, USA, 2010. ACM.

\bibitem[Dughmi et~al.(2014)Dughmi, Han, and Nisan]{dughmi2014sampling}
Shaddin Dughmi, Li~Han, and Noam Nisan.
\newblock Sampling and representation complexity of revenue maximization.
\newblock In \emph{Web and Internet Economics}, volume 8877 of \emph{Lecture
  Notes in Computer Science}, pages 277--291. Springer Intl. Publishing,
  Beijing, China, 2014.

\bibitem[Elkind(2007)]{elkind2007}
Edith Elkind.
\newblock Designing and learning optimal finite support auctions.
\newblock In \emph{Proceedings of the eighteenth annual ACM-SIAM symposium on
  Discrete algorithms}, pages 736--745. SIAM, 2007.

\bibitem[Hartline(2015)]{hartline2015}
Jason Hartline.
\newblock \emph{Mechanism design and approximation}.
\newblock Jason Hartline, Chicago, Illinois, 2015.

\bibitem[Hartline and Roughgarden(2009)]{hartline2009simple}
Jason~D. Hartline and Tim Roughgarden.
\newblock Simple versus optimal mechanisms.
\newblock In \emph{ACM Conf. on Electronic Commerce}, Stanford, CA, USA., 2009.
  ACM.

\bibitem[Huang et~al.(2014)Huang, Mansour, and Roughgarden]{huang2014making}
Zhiyi Huang, Yishay Mansour, and Tim Roughgarden.
\newblock Making the most of your samples.
\newblock abs/1407.2479:\penalty0 1--3, 2014.
\newblock URL \url{http://arxiv.org/abs/1407.2479}.

\bibitem[Kearns and Vazirani(1994)]{KV}
Michael~J. Kearns and Umesh~V. Vazirani.
\newblock \emph{An Introduction to Computational Learning Theory}.
\newblock MIT Press, Cambridge, MA, 1994.

\bibitem[Medina and Mohri(2014)]{medina2014learning}
Andres~Munoz Medina and Mehryar Mohri.
\newblock Learning theory and algorithms for revenue optimization in second
  price auctions with reserve.
\newblock In \emph{Proceedings of The 31st Intl. Conf. on Machine Learning},
  pages 262--270, 2014.

\bibitem[Myerson(1981)]{myerson1981optimal}
Roger~B Myerson.
\newblock Optimal auction design.
\newblock \emph{Mathematics of operations research}, 6\penalty0 (1):\penalty0
  58--73, 1981.

\bibitem[Pollard(1984)]{pollard1984}
David Pollard.
\newblock \emph{Convergence of stochastic processes}.
\newblock David Pollard, New Haven, Connecticut, 1984.

\bibitem[Roughgarden and Schrijvers(2015)]{RS15}
T.~Roughgarden and O.~Schrijvers.
\newblock Ironing in the dark.
\newblock Submitted, 2015.

\bibitem[Roughgarden et~al.(2012)Roughgarden, Talgam-Cohen, and Yan]{RTY12}
Tim Roughgarden, Inbal Talgam-Cohen, and Qiqi Yan.
\newblock Supply-limiting mechanisms.
\newblock In \emph{Proceedings of the 13th ACM Conf. on Electronic Commerce},
  pages 844--861, NY, NY, USA, 2012. ACM.

\bibitem[Valiant(1984)]{valiant1984theory}
Leslie~G Valiant.
\newblock A theory of the learnable.
\newblock \emph{Communications of the ACM}, 27\penalty0 (11):\penalty0
  1134--1142, 1984.

\bibitem[Vapnik and Chervonenkis(1971)]{VC}
Vladimir~N Vapnik and A~Ya Chervonenkis.
\newblock On the uniform convergence of relative frequencies of events to their
  probabilities.
\newblock \emph{Theory of Probability \& Its Applications}, 16\penalty0
  (2):\penalty0 264--280, 1971.

\bibitem[Yao(2015)]{yao2015soda}
Andrew Chi-Chih Yao.
\newblock An n-to-1 bidder reduction for multi-item auctions and its
  applications.
\newblock In \emph{Proceedings of the Twenty-Sixth Annual ACM-SIAM Symposium on
  Discrete Algorithms}, pages 92--109, San Diego, CA, USA., 2015. ACM.

\end{thebibliography}

\end{document}